\documentclass[reqno]{amsart}

\usepackage{ifthen,version}
\newboolean{arxiv-version}
\setboolean{arxiv-version}{true}
\ifthenelse{\boolean{arxiv-version}}{\includeversion{prnt-arxiv}}%
{\excludeversion{prnt-arxiv}}
\ifthenelse{\boolean{arxiv-version}}{\excludeversion{prnt-cready}}%
{\includeversion{prnt-cready}}

\usepackage{booktabs} 
\usepackage{IEEEtrantools}
\usepackage{amssymb,latexsym,amsfonts,amsmath,amsthm}
\usepackage{thmtools, thm-restate}

\usepackage{mdframed,lipsum}

\usepackage{graphicx}

\usepackage{mathrsfs}

\usepackage{graphicx}
\usepackage{paralist}
\usepackage{dsfont}
\usepackage{eso-pic}
\usepackage{epsfig}
\usepackage{mathtools}
\usepackage{epstopdf}

\usepackage{wrapfig}

\newcommand{\R}{{\mathbb{R}}}

\DeclareMathOperator*{\opt}{opt}

\newcommand{\Rwd}{\mathsf{R}}

\newcommand{\N}{{\mathbb{N}}}

\newcommand{\EE}{\mathds{E}}
\newcommand{\PP}{\mathds{P}}

\usepackage[many]{tcolorbox}
\tcbuselibrary{skins}

\newtcolorbox{resp}[1][]{%
	enhanced jigsaw,%
	colback=gray!5!white,%
	colframe=gray!80!black,%
	size=small,%
	boxrule=1pt,%
	halign title=flush center,%
	coltitle=black,%
	breakable,%
	drop shadow=black!50!white,%
	attach boxed title to top left={xshift=1cm,yshift=-\tcboxedtitleheight/2,yshifttext=-\tcboxedtitleheight/2},%
	minipage boxed title=3cm,%
	boxed title style={%
		colback=white,%
		size=fbox,%
		boxrule=1pt,%
		boxsep=2pt,%
		underlay={%
			\coordinate (dotA) at ($(interior.west) + (-0.5pt,0)$);
			\coordinate (dotB) at ($(interior.east) + (0.5pt,0)$);
			\begin{scope}[gray!80!black]
				\fill (dotA) circle (2pt);
				\fill (dotB) circle (2pt);
			\end{scope}
		}%
	},%
	#1%
}

\usepackage{tikz,circuitikz}
\usetikzlibrary{patterns,calc,angles,quotes,automata,positioning,decorations.markings,arrows,intersections,circuits.logic.US,shapes.arrows,shapes.gates.logic.US,backgrounds}
\usetikzlibrary{automata,positioning,decorations.markings,arrows,intersections,calc,shapes}

\colorlet{darkgreen}{green!40!black}
\colorlet{darkblue}{blue!60!black}
\colorlet{darkred}{red!50!black}
\colorlet{safecellcolor}{yellow!5}
\colorlet{goodcellcolor}{green!10}
\colorlet{badcellcolor}{blue!10}

\tikzset{
	>=stealth,
	box state/.style={draw,rounded corners,rectangle,minimum size=8mm},
	prob state/.style={draw,very thick,shape=circle,darkblue,minimum size=3mm,inner sep=0mm},
	node distance=2cm,on grid,auto, initial text=,
	every loop/.style={shorten >=0pt},
	accepting/.style={double distance=1.2pt, outer sep = 0.6pt+\pgflinewidth},
	accepting dot/.style={above=-2.5pt,circle,fill,darkgreen,inner sep=2pt,radius=1pt},
	loop above/.append style={every loop/.append style={out=120, in=60, looseness=5}},
	loop below/.append style={every loop/.append style={out=300, in=240, looseness=5}},
	loop left/.append style={every loop/.append style={out=210, in=150, looseness=5}},
	loop right/.append style={every loop/.append style={out=30, in=330, looseness=5}},
	accepting arc/.style={dashed},
	marked/.style={
		dashed,
		opacity=0.3
	},
	marked on/.style={alt=#1{marked}{}},
}
\ctikzset{bipoles/not port/width=0.75,bipoles/not port/height=0.5}
\ctikzset{bipoles/not port/circle width=0.3}
\ctikzset{tripoles/american nand port/circle width=0.2}

\usepackage{algorithm}
\usepackage{algorithmic}

\usepackage{xspace}

\newcommand{\Nat}{{\mathbb{N}}}

\newcommand{\Real}{{\mathbb{R}}}

\newcommand{\Mm}{\mathcal{M}}
\newcommand{\Gg}{\mathcal{G}}
\newcommand{\Bb}{\mathcal{B}}

\newcommand{\alphabet}{\mathsf{\Sigma}_{\textsf{a}}}

\newcommand{\Lab}{\mathsf{L}}

\newcommand{\until}{\mathbin{\sf U}}

\newcommand{\nex}{\mathord{\bigcirc}}

\newcommand{\seq}[1]{\langle #1 \rangle} 

\def\BibTeX{{\rm B\kern-.05em{\sc i\kern-.025em b}\kern-.08em
		T\kern-.1667em\lower.7ex\hbox{E}\kern-.125emX}}

\newcommand{\mMIN}{\mathrm{Min}}
\newcommand{\mMAX}{\mathrm{Max}}

\newtheorem{theorem}{Theorem}[section]

\newtheorem{problem}[theorem]{Problem}
\newtheorem{proposition}[theorem]{Proposition}

\newtheorem{definition}[theorem]{Definition}

\newtheorem{remark}[theorem]{Remark}

\usepackage{fancyhdr}

\newenvironment{nouppercase}{%
	\renewcommand{\uppercasenonmath}[1]{}}{}

\begin{document}

\begin{abstract}
We propose a compositional approach to synthesize policies for networks of continuous-space stochastic control systems with unknown dynamics using model-free reinforcement learning (RL). 
The approach is based on \emph{implicitly} abstracting each subsystem in the network with a finite Markov decision process with \emph{unknown} transition probabilities, synthesizing a strategy for each abstract model in an assume-guarantee fashion using RL, and then mapping the results back over the original network with \emph{approximate optimality} guarantees. 
We provide lower bounds on the satisfaction probability of the overall network based on those over individual subsystems.  
A key contribution is to leverage the convergence results for adversarial RL (minimax Q-learning) on finite stochastic arenas to provide control strategies maximizing the probability of satisfaction over the network of continuous-space systems. 
We consider \emph{finite-horizon} properties expressed in the syntactically co-safe fragment of linear temporal logic.
These properties can readily be converted into automata-based reward functions, providing scalar reward signals suitable for RL.
Since such reward functions are often sparse, we supply a potential-based \emph{reward shaping} technique to accelerate learning by producing dense rewards. The effectiveness of the proposed approaches is demonstrated via two physical benchmarks including regulation of a room temperature network and control of a road traffic network.
\end{abstract}

\title{{\LARGE{Compositional Reinforcement Learning for Discrete-Time Stochastic Control Systems}}}

\author{{\bf {\large Abolfazl Lavaei$^1$}}}
\author{{\bf {\large Mateo Perez$^2$}}}
\author{{\bf {\large Milad Kazemi$^1$}}}
\author{{\bf {\large Fabio Somenzi$^3$}}}
\author{{\bf {\large Sadegh Soudjani$^1$}}}
\author{{\bf {\large Ashutosh Trivedi$^2$}}}
\author{{\bf {\large Majid Zamani$^{2,4}$}}\\
{\normalfont $^1$School of Computing, Newcastle University, United Kingdom}\\
{\normalfont $^2$Department of Computer Science, University of Colorado Boulder, USA}\\
{\normalfont $^3$Electrical, Computer \& Energy Engineering, University of Colorado Boulder, USA}\\
{\normalfont $^4$Department of Computer Science, LMU Munich, Germany}\\
\texttt{\{abolfazl.lavaei,m.kazemi2,sadegh.soudjani\}@newcatle.ac.uk}
\texttt{\{mateo.perez,fabio,ashutosh.trivedi,majid.zamani\}@colorado.edu}}

\pagestyle{fancy}
\lhead{}
\rhead{}
\fancyhead[OL]{A. Lavaei, M. Perez, M. Kazemi, F. Somenzi, S. Soudjani, A. Trivedi, M. Zamani}

\fancyhead[EL]{Compositional Reinforcement Learning for Discrete-Time Stochastic Control Systems} 
\rhead{\thepage}
\cfoot{}

\begin{nouppercase}
	\maketitle
\end{nouppercase}

\section{Introduction}
As edge-computing coupled with the internet-of-things continue to transform the critical infrastructure (e.g., traffic networks and power grids), there is ever-increasing demands on automatic control synthesis for interconnected networks of continuous-space stochastic systems. Often closed-form models for such physical environments are either unavailable or too complex to be of practical use, rendering model-based design approaches impractical.
Although system identification techniques~\cite{cheng2019end,jagtap2020control} can be employed to learn an approximate model, acquiring accurate models for large and complex systems remains challenging, time-consuming, and expensive. 
Model-free reinforcement learning (RL)~\cite{Sutton18} are sampling-based approach to synthesize controllers that compute the optimal policies without constructing a full model of the system, and hence are asymptotically more space-efficient than model-based approaches. 
\emph{We develop convergent model-free reinforcement learning for the controller synthesis of networks of (unknown) continuous-space stochastic systems against formal requirements.}

Abstraction-based synthesis \cite{HS18_robust,SSoudjani,lahijanian2016iterative,luo2021abstraction,SA13,majumdar2021symbolic} is an effective approach for the synthesis of continuous-space stochastic systems.
The recipe of the abstraction-based synthesis has the following three steps: {\bf abstraction} (an approximation of the dynamics to a finite-state Markov Decision Process or MDP), {\bf policy synthesis} (synthesis of an optimal policy for the abstract model), and {\bf policy transfer} (a translation of the results back to the original system, while establishing bounds on the error due to the abstraction process).
In a previous work~\cite{lavaei2020ICCPS}, we combined the guarantees from the abstraction-based synthesis with the RL convergence guarantees for finite-state MDPs to provide an RL algorithm for MDPs with uncountable state sets while providing convergence guarantees. This approach enabled us to apply model-free, off-the-shelf RL algorithms to compute $\varepsilon$-optimal strategies for continuous-space MDPs with a precision $\varepsilon$ that is defined a-priori and without explicitly constructing finite abstractions.
As these aforementioned abstraction-based synthesis approaches \cite{SSoudjani,lahijanian2016iterative, luo2021abstraction,SA13,majumdar2021symbolic} rely on state-space discretization, they severely suffer from the curse of dimensionality.
To mitigate this issue, \emph{compositional} abstraction-based techniques have been introduced to construct finite abstractions of large systems based on those of smaller subsystems~\cite{lavaei2018ADHSJ,lavaei2018CDCJ,lavaei2019Thesis,Lavaei_Survey,Lavaei_TAC2022}. 
This paper exploits the network structure to scale RL-guided abstraction-based synthesis for networks of continuous-state MDPs. 

\vspace{1em}\noindent\textbf{Contributions.} 
This paper studies RL-based controller synthesis when the environment is a network of systems. 
Instead of analyzing the whole network monolithically, and hence facing the scalability barrier, this paper investigates a compositional approach that solves the optimization problem for each subsystem in the network separately, while considering the other subsystems as adversaries in a two-player game.
In particular, we present a compositional approach to scale RL to synthesize policies against \emph{finite-horizon} specifications in networks of (partially) unknown stochastic systems while providing convergence guarantees. 
Our approach is applicable to uncountable, but bounded, state sets with finite input sets, and requires the knowledge of the Lipschitz constants of the subsystems.
Since the transition probabilities remain unknown, we employ the convergent multi-agent RL \cite{Littma96} to synthesize strategies. 

We utilize a closeness guarantee between probabilities of satisfaction
by subsystems and their implicit finite MDPs (which can be chosen
a-priori), and leverage convergence results of minimax-Q learning
\cite{Littma96} for solving stochastic games on finite MDPs. We
provide, for the first time, a theoretical lower bound on the
probability of satisfaction of finite-horizon properties by the original interconnected continuous-space
stochastic system with unknown dynamics in terms of the bounds computed for subsystems. We also propose a novel potential-based reward shaping technique to produce dense rewards, which is based on the structure of the automata representing the specifications of interest. Finally, we demonstrate our approach on two case studies.

\noindent{\bf Related Work.} A model-free RL framework for
synthesizing policies for unknown, and possibly continuous-state, stochastic systems is
presented
in~\cite{hasanbeig2019reinforcement,yuan2019modular, kazemi2020formal,wang2020continuous}. Our
proposed approaches here differ from the ones
in~\cite{hasanbeig2019reinforcement,yuan2019modular, wang2020continuous} in
two main directions. First, the  proposed approaches
in~\cite{hasanbeig2019reinforcement,yuan2019modular}
provide theoretical guarantees only if the underlying system has finitely many states.
In contrast, we learn
$\varepsilon$-optimal strategies for original continuous-space systems with a-priori defined precision $\varepsilon$.
In addition, we propose a
compositional RL framework for the policy synthesis of \emph{networks} of continuous-space stochastic systems, whereas the results
in~\cite{hasanbeig2019reinforcement,yuan2019modular, kazemi2020formal,wang2020continuous} only
deal with monolithic systems.

Our solution approach is related to our prior work \cite{lavaei2020ICCPS} where we develop RL-guided abstraction-based synthesis approach for continuous-state stochastic control systems. The present work differs from~\cite{lavaei2020ICCPS} in several directions.
First and foremost, the results in~\cite{lavaei2020ICCPS} only deal with monolithic systems and, hence, suffer from the curse of dimensionality. In contrast, we propose here a compositional RL framework for networks of continuous-space stochastic systems by breaking the main synthesis problem into simpler ones. As the second extension, we propose here a multi-level discretization scheme for RL in which the agent learns control policies on a sequence of finer and finer discretizations of the same system. We show that this improves learning efficiency while preserving convergence results. 

Finally, our theoretical results extend the abstraction error analysis of \cite{SA13} from single-player to multi-player controller synthesis, and from monolithic systems to network of systems. In particular, the error quantification is performed for max-min optimizations on local systems and provides a lower bound after composing them in a network (cf. Theorem~\ref{thm:com}).

\section{Preliminaries}
\label{prelim}
We write $\Nat$ and $\Real$ for the set of natural and real numbers.
For a set of $N$ vectors $x_i \in \mathbb R^{n_i}$, $1 {\leq} i {\leq} N$, we write $[x_1;\ldots;x_N]$ to denote the corresponding column vector of
dimension $\sum_i n_i$. We denote by $\mathbf{0}_n$ a column vector of all zeros in $\mathbb R^{n}$.
Given functions $f_i:X_i{\rightarrow} Y_i$, for $1 {\leq} i{\leq} N$, their product $\bigtimes_{i=1}^{N}f_i\!:\bigtimes_{i=1}^{N}X_i{\rightarrow}\bigtimes_{i=1}^{N}Y_i$ is defined as $(x_1,\ldots,x_N) {\mapsto} [f_1(x_1);\ldots;f_N(x_N)]$.
We represent a diagonal matrix with $\sigma_1,\ldots,\sigma_n$ as its entries as $\operatorname{diag}(\sigma_1,\ldots,\sigma_n)$.

A probability space is a tuple $(\Omega,\mathcal F_{\Omega},\mathds{P}_{\Omega})$
where $\Omega$ is the sample space, $\mathcal F_{\Omega}$ is a $\sigma$-algebra on
$\Omega$ comprising subsets of $\Omega$ as events, and $\mathds{P}_{\Omega}$ is
a probability measure that assigns probabilities to events.
We assume that random variables are measurable
functions of the form
$X:\Omega \to S_X$. 
Any random variable $X$ induces a probability measure on  its space
$(S_X,\mathcal F_X)$ as $Prob\{A\} = \mathds{P}_{\Omega}\{X^{-1}(A)\}$ for any
$A\in \mathcal F_X$.
A \emph{discrete probability distribution}, or just distribution, over a
set $X$ is a function $d : X {\to} [0, 1]$ such that 
$\sum_{x \in X} d(x) = 1$.
A topological space $S$ is called a {\it Borel space} if it is homeomorphic to a Borel
subset of a Polish space (\emph{i.e.,} a separable and completely metrizable
space). 
Any Borel space $S$ is assumed to be endowed with a Borel $\sigma$-algebra, which is
denoted by $\mathcal B(S)$. We say that a map $f : S\rightarrow Y$ is measurable
whenever it is Borel measurable. 

\subsection{Discrete-Time Stochastic Control Systems}

We consider networks of stochastic control systems in discrete time where each subsystem, a discrete-time stochastic control system, is defined as follows.	
\begin{definition}
	A {\it discrete-time stochastic control system} (dt-SCS) is a tuple 
	\begin{equation}\label{eq:dt-SCS1}
	\Sigma=\left(X,U,W,\varsigma,f,Y, h\right)\!,
	\end{equation}
	where:
	\begin{itemize}
		\item
		$X\subseteq \mathbb R^n$, a Borel space, is the state set of the
		system. 
		\item 
		$U$ is the \emph{external} input set which is finite; 
		\item 
		$W\subseteq \mathbb R^p$ is the \emph{internal} input set;
		\item
		$\varsigma$ is a sequence of independent and identically distributed random variables from a sample space $\Omega$ to the set
		$\mathcal V_\varsigma$, namely	 
		$\varsigma:=\{\varsigma(k):\Omega\rightarrow \mathcal V_{\varsigma},\,\,k\in\N\}$;
		\item
		$f:X\times U \times W \times \mathcal V_{\varsigma} \rightarrow X$ is a measurable function characterizing the state evolution of $\Sigma$;
		\item
		$Y\subseteq \mathbb R^q$ is the output set;
		\item
		$h: X {\rightarrow} Y$, a measurable function, maps states to outputs. 
	\end{itemize}
	We write $\PP\{f(x, v, w, \cdot) \in B  \mid x, u, w\}$ for the probability that the next state is in $B \in \Bb(X)$ given current state $x {\in} X$, external input $u {\in} U$, internal input $w {\in} W$, when the remaining argument is distributed like the random variables in $\varsigma$.
\end{definition}

The execution of $\Sigma$ from $x(0){\in} X$, and inputs $\{\upsilon(k): \Omega{\rightarrow} U,\,\,k{\in}\mathbb N\}$ and $\{w(k):\Omega{\rightarrow} W,\,\,k{\in}\mathbb N\}$ is described by: 
\begin{equation}\label{Eq_1a1}
\Sigma\!:\left\{\hspace{-1.5mm}
\begin{array}{l}
x(k+1)=f(x(k),\upsilon(k),w(k),\varsigma(k)),\\
y(k)=h(x(k)),\\
\end{array}\right.
\quad k\in\mathbb N.
\end{equation}

We also consider special subclass of dt-SCS, called closed dt-SCS, where the internal inputs are absent, \emph{i.e.,}  when $w(k)=0$, $\forall k\in \N$.
Such systems may also result from considering an interconnection of dt-SCSs  (cf. Definition~\ref{interconnected}).
We represent closed dt-SCS as 
$(X,U,\varsigma,f)$ and its execution can be simplified to 
\begin{equation}\label{eq:dt-SCS}
\Sigma\!:\left\{\hspace{-1.5mm}
\begin{array}{l}
x(k+1)=f(x(k),\upsilon(k),\varsigma(k)),\\
y(k)=h(x(k)),\\
\end{array}\right.
\quad k\in\mathbb N.
\end{equation}
For a closed dt-SCS, we write $\PP\{f(x, v, \cdot) \in B  \mid x, u\}$ for the probability that the next state is in $B$ given the current state $x \in X$ and input $u \in U$. We call a non-closed dt-SCS open.
When clear, we drop the open or closed specifier. 

\begin{definition}[Network of dt-SCS]\label{interconnected}
	Let $\Sigma_i=(X_i,U_i,W_i,\varsigma_i, f_i, Y_i,h_i)$, for $ 1 \leq i \leq N$, be a family of $N$ open dt-SCS.
	The network of $\seq{\Sigma_i}_{1\leq i \leq N}$ is defined by
	the interconnection map $g:\bigtimes_{i=1}^{N}Y_i\to\bigtimes_{i=1}^{N}W_i$,
	and gives rise to a closed dt-SCS 
	$\mathcal{I}_g(\Sigma_1,\ldots,\Sigma_N) = (X,U,\varsigma, f)$, where 
	$X:=\bigtimes_{i=1}^{N}X_i$,  $U:=\bigtimes_{i=1}^{N}U_i$, and
	$f:=\bigtimes_{i=1}^{N}f_{i}$, subjected to the following interconnection constraint: 
	\begin{align}
	[w_1;\ldots;w_N]=g(h_{1}(x_1),\dots,h_{N}(x_N)).
	\end{align}
\end{definition}

An example of the interconnection of two stochastic control subsystems $\Sigma_1$ and $\Sigma_2$ is illustrated in Fig.~\ref{Figg1}.

\begin{figure}[t]
	\begin{center}
		\begin{tikzpicture}[>=latex']
		\tikzstyle{block} = [draw, 
		thick,
		rectangle, 
		minimum height=.8cm, 
		minimum width=1.5cm]
		
		\draw[dashed] (-1.7,-2.2) rectangle (1.7,.7);
		
		\node[block] (S1) at (0,0) {$\Sigma_1$};
		\node[block] (S2) at (0,-1.5) {$\Sigma_2$};
		
		\draw[->] ($(S1.east)+(0,0.25)$) -- node[very near end,above] {$x_{1}$} ($(S1.east)+(1.5,.25)$);
		\draw[<-] ($(S1.west)+(0,0.25)$) -- node[very near end,above] {$\upsilon_{1}$} ($(S1.west)+(-1.5,.25)$);
		
		\draw[->] ($(S2.east)+(0,-.25)$) -- node[very near end,below] {$x_{2}$} ($(S2.east)+(1.5,-.25)$);
		\draw[<-] ($(S2.west)+(0,-.25)$) -- node[very near end,below] {$\upsilon_{2}$} ($(S2.west)+(-1.5,-.25)$);
		
		\draw[->] 
		($(S1.east)+(0,-.25)$) -- node[very near end,above] {$h_{1}$}
		($(S1.east)+(.5,-.25)$) --
		($(S1.east)+(.5,-.5)$) --
		($(S2.west)+(-.5,.5)$) --
		($(S2.west)+(-.5,.25)$) -- node[very near start,below] {$w_{2}$}
		($(S2.west)+(0,.25)$) ;
		
		\draw[->] 
		($(S2.east)+(0,.25)$) -- node[very near end,below] {$h_{2}$} 
		($(S2.east)+(.5,.25)$) --
		($(S2.east)+(.5,.5)$) --
		($(S1.west)+(-.5,-.5)$) --
		($(S1.west)+(-.5,-.25)$) -- node[very near start,above] {$w_{1}$}
		($(S1.west)+(0,-.25)$) ;
		\end{tikzpicture}
	\end{center}
	\caption{Interconnection $\mathcal{I}(\Sigma_1,\Sigma_2)$ of stochastic control subsystems $\Sigma_1$ and $\Sigma_2$.}
	\label{Figg1}
\end{figure}

\subsection{Stochastic Games and Markov Decision Processes}
The semantics of an open dt-SCS $\Sigma$ can be naturally expressed as a stochastic game~\cite{Filar97} between two players---player Max (the control), who controls the external inputs $U$,  and player Min (the adversary), who controls the internal inputs.
We assume that the adversary is more powerful than the controller in that the adversary can see the choices of the controller at every step.
From the control synthesis perspective, this view results in a cautious controller with a pessimistic view of the environment. 
On the other hand, a strategy computed in this manner also works against the weaker adversary. 

A \emph{stochastic game arena} (SGA) is a tuple $\Gg = (S, A, T, S_{\mMAX}, S_{\mMIN})$ where:
$S$ is (potentially uncountable) state set;
$A$ is the set of actions and $A(s)$ is the set of actions enabled at $s \in S$;
$T: S \times A \times \mathcal{B}(S) {\to} [0,1]$ is a conditional stochastic kernel that, for $(s, a) {\in} S{\times} A$, assigns a probability measure $P(\cdot | s, a)$ on the measurable space $(S,\mathcal B(S))$.
$S_{\mMAX} \subseteq S$ and $S_{\mMIN} \subseteq S$  form a partition of $S$ into the set of states controlled by players Max and Min, respectively. 
For the stochastic kernel $T$, state $s \in S$, action $a \in A$, and set $B \in \Bb(S)$, we write $T(B \mid s, a)$ for $T(s, a, B)$.
We say that a SGA is finite, if both $S$ and $A$ are finite.
For finite SGAs, the transition function $T(s, a, \cdot)$ is a discrete probability distribution for every $s \in S$ and $a \in A$. For a finite SGA, we write $T(s' \mid s, a)$ for $T(s, a, s')$ for all $s, s' \in S$ and $a \in A$.
An SGA is an MDP if $S_\mMIN = \emptyset$ and represent it as $(S, A, T)$.
An MDP $(S, A, T)$ is finite if both $S$ and $A$ are finite. 

\begin{definition}[dt-SCS: Semantics]\label{dt-scs-semantics}
	An open dt-SCS $\Sigma=(X, U, W, \varsigma, f, Y, h)$ can be interpreted as an SGA
	$\mathcal{G}_\Sigma=(S, A, T, S_\mMAX, S_\mMIN)$, where 
	\begin{itemize}
		\item $S = X \cup (X \times U)$ such that $S_\mMAX = X$ and $S_\mMIN = X{\times} U$; 
		\item $A = U \cup W$ such that $A(s) = U$ for $s \in X$ and for $A(s) = W$ for $s \in X \times U$;
		\item $T: S \times A  \times \Bb(S) \to [0,1]$ such that 
		\begin{itemize}
			\item 
			$T((x, u) \mid x, u) = 1$ 
			for $x \in S_\mMAX$ and $u \in U$ 
			\item 
			$T(S\backslash\{(x, u)\} \mid x, u) = 0$ 
			for $x \in S_\mMAX$ and $u \in U$ 
			\item
			$T(B \mid (x, u), w) = \mathbb P \{ f(x, v, w, \cdot) \in B  \mid x, u, w\}$
			for all $(x, u) \in S_\mMIN$, $w \in W$, and all $B \in \Bb(X)$.
		\end{itemize}
	\end{itemize}
	
	Similarly, a closed dt-SCS $\Sigma = (X, U, \varsigma, f)$ can be \emph{equivalently} represented as an MDP
	$\Mm_\Sigma = (S=X, A=U, T)$ where $T:S \times A \times \Bb(S) \rightarrow[0,1]$
	such that for all $x \in X$, $u \in U$, and $B \in \Bb(X)$, we have that 
	\[
	T(B \mid x, u)	= \mathbb P \Big\{ f(x, u, \cdot) \in B \,\big|\, x, u\Big\}.
	\]
	Abusing notation, we write $\Sigma$ for its SGA $\Gg_\Sigma$ or MDP $\Mm_\Sigma$.
\end{definition}

The objective in an SGA is to determine a policy---a decision rule for every step to choose the next action---for both players that optimizes a given objective.
Although the policy could in general be history-dependent or randomized, w.l.o.g. we only consider memoryless, deterministic policies~\cite{ding2013stochastic}. We call such policies Markov policies.
For an SGA $\Gg$, a \emph{Markov policy} $\rho$ of player Max is a sequence $(\rho_0, \rho_1,\rho_2,\ldots)$ where each rule $\rho_k: S_\mMAX \to A$, for $k \in \Nat$, is a universally measurable function such that $\rho_n(s) \in A(s)$ for all $s \in S$. 
Similarly, a Markov policy $\xi$ of player Min is a sequence $(\xi_0, \xi_1,\xi_2,\ldots)$ where $\xi_k: S_\mMIN \to A$, for $k \in \Nat$, is a universally measurable function such that $\xi_k(s) \in A(s)$ for all $s_\mMIN \in S$.
We write $\Pi^\Gg_\mMAX$ and $\Pi^\Gg_\mMIN$ for the set of all Markov policies of players Max and Min, respectively.
For an MDP $\Mm$ we write $\Pi^\Mm$ for the set of policies.
We omit the superscripts $\Mm$ and $\Gg$ when clear from the context.

Any pair of Markov policies $\rho \in \Pi^\Gg_\mMAX$ and $\xi \in \Pi^\Gg_\mMIN$ and initial state $s \in S$ characterize a unique stochastic process over sequences of states and actions. We write $\Gg^{s}_{\rho,\xi}$ for this stochastic process and write $S_k$ and $A_k$ for the random variables corresponding to the state and action at time step $k \in \Nat$. 
We write $\EE^{s}_{\rho,\xi}[\cdot]$ for the expected value of a random variable for the stochastic process $\Gg^{s}_{\rho,\xi}$.
If we also condition on the initial action $a \in A$, we write $\Gg^{s,a}_{\rho,\xi}$ and $\EE^{s,a}_{\rho,\xi}[\cdot]$.
Similarly, we write $\Mm^{s}_{\rho}$ for the stochastic process of an MDP $\Mm$ with initial state $s$ and policy $\rho$, and $\EE^{s}_{\rho}[\cdot]$ for the corresponding expectation.

\subsection{Finite-Horizon Specifications}\label{property}
This paper deals with a fragment of linear temporal logic formulae known as \emph{syntactically co-safe linear temporal logic} (scLTL), in which the negation operator ($\neg$) only occurs before atomic propositions\cite{KupfermanVardi2001} characterized by the grammar:
\[
\varphi := p \mid \neg p \mid \varphi \vee \varphi \mid  \varphi \wedge \varphi \mid \nex \varphi \mid \varphi \until \varphi.
\]
We refer the reader to \cite{baier2008principles} for the syntax and semantics of LTL.
We denote the language of finite words associated with an scLTL formula
$\varphi$ by $\mathcal L_f(\varphi)$. 

A {\it deterministic finite automaton} (DFA) is a tuple
$\mathcal A \!=\! (Q, \alphabet, \mathsf{t}, q_0, F_{\textsf{a}})$,
where $Q$ is a finite set of states, $\alphabet$ is a finite alphabet, $\mathsf{t}: Q \times \alphabet \to Q$ is a transition function,
$q_0 \in Q$ is the initial state, and  $F_{\textsf{a}} \subseteq Q$ is a set of accepting states.
We write $\lambda$ for the empty word and $\alphabet^*$ for
the set of finite strings over $\alphabet$.
The extended transition function $\mathsf{\hat t}: Q \times \alphabet^* \to Q$
is defined as:
\begin{equation*}
\mathsf{\hat t}(q, \bar w) =
\begin{cases}
q, \!&~ \text{ if $\bar w = \lambda$},\\
\mathsf{t}(\mathsf{\hat t}(q, x), a), \!& ~\text{ if $\bar w = xa$ for $x\in\mathsf{\Sigma}_{\textsf{a}}^*$ and \!$a \!\in\! \mathsf{\Sigma}_{\textsf{a}}$.} 
\end{cases}
\end{equation*}
The language accepted by $\mathcal A$ is 
$\mathcal{L}(\mathcal A) \!=\! \{\bar w\!\in\! \alphabet^* \!\:|\:\! \mathsf{\hat t}(q_0, \bar w)) \!\in\! F_{\textsf{a}}\}$.
For verification and synthesis purposes, an scLTL formula $\varphi$ can be compiled into a \emph{deterministic finite automaton} $\mathcal A_{\phi}$
such that $\mathcal L_f(\varphi) = \mathcal L(\mathcal A_{\varphi})$ \cite{KupfermanVardi2001}.
The resulting DFA has an accepting state whose out-going transitions are all self-loops.
Such a DFA is known as a {\it co-safety automaton}.

\begin{remark}
	We assume that the DFA $\mathcal{A}_\varphi$ for an scLTL property $\varphi$ is a co-safety one.
	Some formulae of scLTL describe finite-horizon properties.  For example, $p \vee \nex\nex q$ only requires checking the first three letters of a word. Other properties, like $p \until q$ are satisfied by finite words of arbitrary length.
	We can adjoin a finite time horizon $\mathcal{T}$ to a formula $\varphi$ and stipulate that an infinite word satisfies $(\varphi,\mathcal{T})$ if it has a prefix of length at most $\mathcal{T}{+}1$ that is in $\mathcal{L}(\mathcal{A}_\varphi)$.
\end{remark}

\section{Problem Definition}
\label{key-results1}
We study a \emph{compositional approach} for the controller synthesis of \emph{networks} of unknown continuous-space stochastic systems under finite-horizon specifications. 
We apply a model-free \emph{two-player} RL in an assume-guarantee fashion and \emph{compositionally} compute policies over finite horizons for original networks without explicitly constructing their finite abstractions. We then propose a lower bound for the optimality of synthesized controllers when applied to the interconnected system based on those of individual controllers.

\subsection{Compositional Controller Synthesis}
In order to provide any formal guarantee, we assume that the system in \eqref{Eq_1a1}
is
Lipschitz-continuous with respect to states and internal inputs with constants $\mathscr{H}_x$ and $\mathscr{H}_w$, respectively, which are the only required
knowledge about the system. An alternative way of having the Lipschitz constants $\mathscr{H}_x$ and $\mathscr{H}_w$ is to
estimate it from sample trajectories of $\Sigma$;
we refer the interested reader to~\cite[equation (3)]{lavaei2020ICCPS} for more details.

\begin{problem}[Compositional Synthesis]
	\label{problem1}
	Let $\varphi_i$ be finite-horizon objectives and
	$\Sigma_i=(X_i,U_i,W_i,\varsigma_i,f_i,Y_i, h_i)$ continuous-space
	subsystems, where $f_i,h_i$ and distribution of $\varsigma_i$ are
	unknown, but Lipschitz constants $\mathscr{H}_{x_i}$ and
	$\mathscr{H}_{w_i}$ are known, for $1 \leq i \leq N$. Synthesize
	Markov policies that satisfy $\varphi_i$ over $\Sigma_i$
	with probabilities within a-priori defined thresholds
	$\varepsilon_i$ from unknown optimal ones and establish a lower bound on the satisfaction probability for the interconnected system.
\end{problem}

To present our solution, we first report the following result~\cite{SA13} to show the closeness
between a continuous-space subsystem $\Sigma$ and its finite
abstraction $\widehat\Sigma$ in a probabilistic setting.
We then leverage it to provide a
two-player stochastic game RL-based solution to Problem~\ref{problem1}. Note that
all (Markov) policies for $\widehat{\Sigma}$ are also (Markov) policies for
$\Sigma$, \emph{i.e.,} $\widehat{\Pi}_{\mMAX} \subseteq \Pi_{\mMAX}$ and
$\widehat{\Pi}_{\mMIN} \subseteq \Pi_{\mMIN}$.

\begin{theorem}\label{thm:key-thm1} 
	Let $\Sigma\!=\!(X,U,W,\varsigma,f,Y,h)$ be a
	continuous-space subsystem and
	$\widehat\Sigma\!=\!(\hat X, \hat U,\hat W,\varsigma,\hat f,\hat Y,\hat
	h)$ be its finite abstraction. For a given finite-horizon objective $\varphi$,
	initial state $x \in X$, and Markov policy pair
	$(\rho,\xi) \in \widehat{\Pi}_{\mMAX} \times \widehat{\Pi}_{\mMIN}$ for the closed-loop
	$\widehat\Sigma$ (with solution process denoted by $\widehat\Sigma^x_{(\rho,\xi)}$), the
	closeness between $\Sigma$ and $\widehat\Sigma$ in terms of satisfaction probability is given by
	\begin{equation}\label{eq:metric_lit1}
	|\mathbb P(\Sigma^x_{(\rho,\xi)} \models \varphi) - \mathbb
	P(\widehat\Sigma^{\hat x}_{(\rho,\xi)} \models \varphi)| \le \varepsilon,
	~~~\text{with}~~~\varepsilon \!:=\! \mathcal T \mathscr{L}(\delta \mathscr{H}_{x} \!+\!
	\mu\mathscr{H}_{w}),
	\end{equation}
	where $\hat x = Q_x(x)$ is a uniform grid quantization map with distance $\delta$ between the grid lines, $\mathcal T$ is the time horizon, $\delta$ and $\mu$ are
	respectively state and internal input discretization parameters,
	$\mathscr{H}_{x}$ and $\mathscr{H}_{w}$ are respectively Lipschitz
	constants of the stochastic kernel w.r.t. states and
	internal inputs, and $\mathscr{L}$ is the Lebesgue measure of (bounded) state set. The selection of $\delta$ and $\mu$ offers a trade-off---decreasing $\delta$ and $\mu$ leads to a finer abstraction with a smaller discretization error, but at the cost of increasing the size of the finite abstraction.
\end{theorem}

If the state set is unbounded, we assume that the specification requires the system to stay in a safe bounded subset of the state set, and this bounded subset can be used in the above theorem.
Next, we consider networks of stochastic control subsystems and provide a lower bound for the \emph{satisfaction probability} of synthesized controllers when applied to the network based on those of individual controllers applied to subsystems as in Theorem~\ref{thm:key-thm1}.

\begin{theorem}\label{thm:com}
	Let $\Sigma\!=\!(X,U,\varsigma,f)$ be an interconnected continuous-space stochastic control system and $\widehat\Sigma=(\hat X, \hat U,\varsigma,\hat f)$ be its finite abstraction. For a given finite-horizon objective $\varphi=\varphi_1\wedge\varphi_2\wedge\ldots\wedge\varphi_N$, we have
	\begin{equation}
	\label{eq:metric_com}
	\mathbb P(\Sigma^x_{\rho^*}\models\varphi) \ge
	\prod_{i=1}^N \min_{\xi_i \in \Pi^{\Sigma_i}}\mathbb P\left((\widehat\Sigma_i)^{x_i}_{(\rho_i^*, \xi_i)}\models\varphi_i\right){-} \frac{1}{2}[(1{+}\varepsilon)^N\!-\!(1{-}\varepsilon)^N],
	\end{equation}
	where $\varepsilon{:=}\max_{i}\varepsilon_i$, $x=[x_1;\ldots;x_N] {\in} X$ is the initial state of $\Sigma$, and $\rho^* \in \Pi^\Sigma$ is the policy composed of the optimal policies $\rho^*_i \in \Pi^{\Sigma_i}_\mMAX$ for objectives $\varphi_i$.
\end{theorem}
\begin{proof}
	We have
	\begin{align*}
	\mathbb P(\Sigma^x_{\rho^*}\models\varphi) \ge
	\prod_{i=1}^N \min_{\xi_i \in \Pi^{\Sigma_i}}\mathbb P((\Sigma_i)^{x_i}_{(\rho_i^*, \xi_i)}\models\varphi_i).
	\end{align*}
	This inequality holds due to the dynamic programming (DP) formulation of $\mathbb P(\Sigma^x_{\rho^*}\models\varphi)$ \cite{SAM17} and the following property of the Bellman operator used in each iteration of DP:
	\begin{align*}
	& \int_{X_1} \int_{X_2} V_1(\bar x_1)V_2(\bar x_2) T_1(d\,\bar x_1\,|\, x_1,x_2) T_2(d\,\bar x_2\,|\, x_1,x_2)\\
	&  =  \int_{X_1} V_1(\bar x_1) T_1(d\,\bar x_1\,|\, x_1,x_2) \int_{X_2} V_2(\bar x_2)  T_2(d\,\bar x_2\,|\, x_1,x_2)\\
	& \ge \min_{w_1}\!\int_{X_1}\!\!\! V_1(\bar x_1) T_1(d\,\bar x_1|x_1,w_1)\!\min_{w_2}\!\int_{X_2}\!\!\! V_2(\bar x_2)  T_2(d\,\bar x_2| w_2,x_2),
	\end{align*}
	where $T_i$ is the stochastic kernel of $\Sigma_i$ and $V_i$ is the value function used in each iteration of DP, for $i\in\{1,2\}$. This inequality allows us to consider the effect of other subsystems in the worst case and get a lower bound on the solution.  
	
	Since for $i=1, 2, \ldots, N$,
	\begin{align*}
	& \left|\min_{\xi_i \in \Pi^{\Sigma_i}}\mathbb P((\Sigma_i)^{x_i}_{(\rho_i^*, \xi_i)}\models\varphi_i)-\min_{\xi_i \in \Pi^{\Sigma_i}}\mathbb P((\widehat \Sigma_i)^{x_i}_{(\rho_i^*, \xi_i)}\models\varphi_i)\right|\le \varepsilon_i.
	\end{align*}
	one can write 
	\begin{align*}
	\mathbb P&(\Sigma^x_{\rho^*}\models\varphi) 
	\ge \prod_{i = 1}^N( \min_{\xi_i \in \Pi^{\Sigma_i}}\mathbb P((\widehat \Sigma_i)^{x_i}_{(\rho_i^*, \xi_i)}\models\varphi_i)\!-\!\varepsilon_i)\\
	& \ge\! \prod_{i = 1}^N \min_{\xi_i \in \Pi^{\Sigma_i}}\mathbb P((\widehat \Sigma_i)^{x_i}_{(\rho_i^*, \xi_i)}\models\varphi_i)
	-\!  \left[{N\choose 1}\epsilon\!+\!{N\choose 3}\epsilon^3\!+\!\ldots\right]\\
	&= \prod_{i = 1}^N \min_{\xi_i \in \Pi^{\Sigma_i}}\mathbb P((\widehat \Sigma_i)^{x_i}_{(\rho_i^*, \xi_i)}\models\varphi_i)
	-\!\frac{1}{2}[(1+\epsilon)^N\!-\!(1-\epsilon)^N],
	\end{align*}
	where $\epsilon=\max_{i}\epsilon_i$ and it completes the proof.
\end{proof}

\begin{remark}
	Note that there is a trade-off between scalability and conservatism in our compositional framework. Our compositional technique significantly mitigates the curse of dimensionality problem due to the state set discretization in RL. On the downside, the overall probabilistic guarantee for interconnected systems in~\eqref{eq:metric_com} is computed based on the multiplication of probabilistic guarantees for individual subsystems, which makes the results conservative.
\end{remark}

\section{Compositional Reinforcement Learning}\label{sec:rl1}

Given Theorem~\ref{thm:key-thm1}, subsystems $\Sigma_i$, properties
$\varphi_i$, and time horizon $\mathcal T$,
Theorem~\ref{thm:key-thm1} provides the error bound for discretization parameters $\delta_i$ and $\mu_i$ if Lipschitz constants $\mathscr{H}_{x_i}$ and $\mathscr{H}_{w_i}$ are known.
This observation enables us to use minimax-Q learning on each underlying discrete game motivated by Theorem~\ref{thm:com} without explicitly constructing the abstraction by restricting observations of the reinforcement learner to the closest representative point in the discrete set of states. In general, the optimal policy for each game depends on the state of the neighboring systems. To remove this dependence, we construct an underapproximation of the game by allowing the neighboring systems to produce arbitrary state sequences---a worse-case assumption about the dynamics. In this way, the maximizing player no longer requires knowledge of the neighboring states by assuming that the neighboring states will always be the worse case. We denote this new underapproximated game by $\widehat\Sigma_{\delta_i}$. Note that the results we show remain true if one uses the original game instead---or any underapproximation to the original game---and we perform this transformation to readily compute a \emph{fully}-decentralized strategy.

In this section we discuss a construction of reward machine from the specifications such that any convergent RL algorithm for two-player stochastic games where players are optimizing such reward signals converges to the value of $\widehat\Sigma_{\delta_i}$. 
These values can then be used to construct $2\varepsilon_i$-optimal strategy for the concrete dt-SCS $\Sigma_i$. Combining these policies may not result in an optimal distributed policy. Instead, we use Theorem~\ref{thm:com} to provide guarantees about its performance via a lower bound.
The proposed solution is summarized below.

\begin{theorem}\label{convergence}
	Let $\varphi_i$ be given finite-horizon properties, $\varepsilon_i>0$, and let 
	$\Sigma_i=(X_i,U_i,W_i,\varsigma_i,f_i,Y_i, h_i)$ be continuous-space
	subsystems, where $f_i,h_i$ and distribution of $\varsigma_i$ are
	unknown, but Lipschitz constants $\mathscr{H}_{x_i}$ and
	$\mathscr{H}_{w_i}$ are known for $i\in\{1,\dots,N\}$. For discretization parameters $\delta_i$ and $\mu_i$ satisfying 
	\[
	\mathcal T_i \mathscr{L}_i(\delta_i \mathscr{H}_{x_i} + \mu_i\mathscr{H}_{w_i})\leq\varepsilon_i,
	\]
	a convergent model-free reinforcement learning algorithm (\emph{e.g.,} minimax-Q learning~\cite{Littma96}) for two-player stochastic games over $\widehat\Sigma_{\delta_i}$
	converges to a
	$2\varepsilon_i$-optimal strategy for subsystems $\Sigma_i$. Accordingly, one can compute a lower bound for the satisfaction probability for the interconnected system based on~\eqref{eq:metric_com}.
\end{theorem}

In Section~\ref{sec:product}, we discuss how to construct rewards from the finite horizon specifications and discuss an approach to make them more dense. 
Section~\ref{subsec:dis} presents our approach to accelerate RL by selecting multi-level discretization of the state set.
However, before we present these results, we recall convergence results for two-player reinforcement learning on stochastic game arenas.

\subsection{RL for Stochastic Games: The Minimax-Q Algorithm}
\label{subsec:RL}
A (discounted) stochastic game is characterized by a pair ($\mathcal{G}, \Rwd)$ consisting of a finite SGA $\Gg$ and reward function $\Rwd: S {\times} A {\times} S \to \Real$.
From initial state $s_0$, the game evolves by having the player that controls $s_k$ at time step $k$ select an action $a_{k+1} \in A(s_k)$. The state then evolves under distribution $T(\cdot \mid s_k, a_k)$ resulting in a next state $s_{k+1}$ and reward $r_{k+1} = \Rwd(s_k, a_k, s_{k+1})$.
Given a discount factor $\gamma \in [0,1)$, the payoff (from player Min to player Max) is the $\gamma$-discounted sum of rewards, \emph{i.e.,}
$\sum_{k=0}^\infty r_{k+1}\gamma^k$.
The objective of player Max is to maximize the expected payoff, while the objective of player Min is the opposite.

A policy $\rho_* \in \Pi_{\mMAX}$ is optimal if it maximizes
\begin{equation*}
\inf_{\xi \in \Pi_{\mMIN}} \EE^s_{\rho, \xi} \Big[\sum_{k=0}^\infty \Rwd(S_k, A_{k+1}, S_{k+1}) \gamma^k\Big],
\end{equation*}
which is sum of rewards under the worst policy of player Min.
The optimal policies for player Min are defined analogously.
The goal of RL is to compute optimal policies for both players with samples from the game, without \emph{a priori}
knowledge of the transition probability and rewards.
The RL solves this by learning a state-action value function, called $Q$-values, defined as 
\[Q_{\rho,\xi}(s,a) = \EE^{s,a}_{\rho,\xi} \Big[ \sum_{k=0}^\infty
\Rwd(S_k, A_{k+1}, S_{k+1}) \gamma^k \Big]\!,
\]
where $\rho \in \Pi_{\mMAX}$ and $\xi \in \Pi_\mMIN$.
Let  
$Q_*(s,a) = \sup_{\rho \in \Pi_\mMAX} \inf_{\xi \in \Pi_\mMIN}  Q_{\rho,\xi}(s, a)$ be the optimal value.
Given $Q_*(s,a)$, one can extract the policy for both players by
selecting the maximum value action in states controlled by player Max and the
minimum value action in states controlled by player Min. 
The following Bellman optimality equations characterize the optimal solutions and forms the basis for computing the Q-values by dynamic programming:
\[
Q_*(s, a) {=} \sum_{s' \in S} T(s' | s, a) {\cdot} \Big( \Rwd(s', a, s) {+} \gamma \cdot \opt_{a'\in A(s')} Q_*(s',a')\Big),
\]
where $\opt$ is $\max$ if $s' \in S_\mMAX$ and $\min$ if $s' \in S_\mMIN$.
Minimax-Q
learning~\cite{Littma94} estimates the dynamic programming update from the
stream of samples by performing the following update at each time step: 
\begin{align*} 
Q(s(k),a(k)) &:=  (1-\alpha_k) Q(s(k), a(k))\! +\!\alpha_k(r(k+1) \!+ \!\gamma \!\!\opt_{a' \in A(s(k+1))}\!\! Q(s(k+1),a')),
\end{align*}
where $\alpha_k \in (0,1)$, a hyperparameter, is the learning rate at time step $k$.
The Minimax-Q algorithm produces the
controller directly, without producing estimates of the unknown
system dynamics: it is \emph{model-free}. 
Moreover, it reduces to classical Q-learning~\cite{watkins1989learning} for MDPs, \emph{i.e.,} when $S_{\mMIN} = \emptyset$.

\begin{theorem}[Minimax-Q Learning\cite{Littma96}]
	The minimax-Q learning algorithm converges to the unique fixpoint $Q_*(s,a)$ if $r(k)$ is bounded, the learning rate satisfies the Robbins-Monro conditions, i.e., 
	$\sum_{k=0}^\infty \alpha_k = \infty \text{  and  } \sum_{k=0}^\infty \alpha_k^2 < \infty$,
	and all state-action pairs are seen infinitely often.
\end{theorem}

For objectives with bounded horizon (i.e., when $r(k) {=} 0$ for all $k {>} N$ for some fixed $N {>} 0$), this convergence result holds even for undiscounted case $\gamma {=} 1$.

\subsection{Reward Machines}\label{sec:product}
In this subsection, we remove indices $i$ for the sake of simple presentation.
Let us fix a subsystem $\Sigma$, a discretization parameter $\delta$, the abstract SGA $\Gg = \widehat\Sigma_{\delta} = (S, A, T, S_{\mMAX}, S_{\mMIN})$ and its specification ${\mathcal A} = \mathcal A_{\varphi}$.
Recall that the DFA $\mathcal{A}_\varphi$ corresponding to a finite-horizon specification $\varphi$ has the property that there is a unique accepting state and
all out-going transitions from that state are self-loops.
Our goal is to introduce a reward function for $\Gg$ such that strategies maximizing it maximize the probability of satisfaction of ${\mathcal A}$.

While it is convenient to express simple specifications directly as Markovian reward signals $R: S\times A \times S \to \Real$, it is difficult to express complex requirements as  Markovian rewards.
A recent trend~\cite{icarte2018using} is to use finite-state machines that read the observation sequences of the MDP to produce reward: these machines are called reward machines.
A {\it reward machine} (RM) is a tuple
$\mathcal A_R = (Q, \alphabet, \mathsf{t}, q_0, \Rwd)$,
where $Q$ is a finite set of states, $\alphabet$ is a finite alphabet, 
$\mathsf{t}: Q \times \alphabet \to Q$ is a transition function,
$q_0 \in Q$ is the initial state,
and  $\Rwd: Q \times \alphabet \times Q \to \Real$ is a reward function. 
The optimal value for a stochastic game $\Gg = (S, A, T, S_{\mMAX}, S_{\mMIN})$ with rewards expressed as an RM   $\mathcal A_R = (Q, \alphabet, \mathsf{t}, q_0, \Rwd)$ can be achieved by computing optimal values for \emph{product stochastic game} $(\Gg\times {\mathcal A}_R) = ((S^\times, A^\times, T^\times, S^\times_{\mMAX}, S^\times_{\mMIN}), \Rwd^\times)$ where  
\begin{itemize}
	\item $S^\times = S \times Q$, $A^\times = A$, 
	$S^\times_\mMIN = S_\mMIN \times Q$, $S^\times_\mMAX = A_\mMAX \times Q$,
	\item 
	$T^\times: S^\times\times A^\times \times S^\times \to [0, 1]$ is such that for $(x, q), (x', q') \in S^\times$ and $v \in A$, 
	\[
	T^\times((x, q), v, (x', q')) = \begin{cases}
	T(x, v, x') & \text{if $q' = \mathsf{t}(q, \Lab(x))$}\\
	0, & \text{otherwise,}
	\end{cases}
	\]
	\item 
	$\Rwd: S^\times \times A^\times \times S^\times \to \Real$ is defined as 
	$
	\Rwd^\times((x, q), v, (x', q')) = \Rwd(q, \Lab(x), q').
	$
\end{itemize}

The following proposition states the correctness of the reduction to the aforementioned reward machine. The correctness is straightforward~\cite{Courco95} since ${\mathcal A}$ is a deterministic finite automaton with a sink accepting state. 
\begin{proposition}
	\label{prop:product-mdp}
	Given a DFA $\mathcal A  = (Q, \mathsf{\Sigma}_{\textsf{a}}, \mathsf{t}, q_0, F_{\textsf{a}})$ capturing the finite-horizon specification $\varphi$,  we construct the RM 
	$\mathcal A_\Rwd  = (Q, \mathsf{\Sigma}_{\textsf{a}}, \mathsf{t},
	q_0, \Rwd)$ with the same structure as ${\mathcal A}$ with a reward governed by the accepting states, \emph{i.e.,} $\Rwd(q, a, q') = 1$ if $q' \in F_{\textsf{a}}$ for all $q, q' \in Q$ and $a \in \mathsf{\Sigma}_{\textsf{a}}$.
	Since $\mathcal{A}$ is a deterministic automaton,
	an expected reward-optimal policy (for discount factor $\gamma = 1$) in  $(\Gg \times {\mathcal A}_\Rwd)$ for a given player characterizes an optimal policy in $\Gg$ to satisfy $\varphi$.
\end{proposition}

While the RM ${\mathcal A}_\Rwd$ constructed above is correct, the reward signals are quite sparse (can be received when the product SGA visits the unique accepting state).
The sparsity of reward signals is known to result in slow learning~\cite{Sutton18}. 
Inspired from~\cite{lavaei2020ICCPS}, we use a ``shaped'' reward function $\Rwd_\kappa$
(parameterized by a hyper-parameter $\kappa$) such that for suitable values of
$\kappa$, optimal policies for $\Rwd_\kappa$ are the same as optimal policies for
$\Rwd$, but unlike $\Rwd$ the function $\Rwd_\kappa$ is more frequent.

The function $\Rwd_\kappa$ is defined based on the structure of RM ${\mathcal A}_\Rwd$.
Let $d(q)$ be the minimum distance of the state $q$ to the unique accepting
state $q_F$.
Let $d_{\texttt{max}} = 1 + \max_{q\in Q} \{ d(q) \::\: d(q) < \infty \}$.
If there is no path from $q$ to $q_F$, let $d(q)$ be equal to $d_{\texttt{max}}$.
We define the \emph{potential function} $P: \mathbb{N} \to \mathbb{R}$ as the following:
\[
P(d) =
\begin{cases}
\kappa \frac{d - d(q_0)}{1 - d_{\texttt{max}}}, &\quad \text{ for $d > 0$},\\
1, &\quad \text{ for $d = 0$},
\end{cases}
\]
where $\kappa$ is a constant hyper-parameter.
Note that the potential of the initial state $P(d(q_0))$ is $0$ and the potential $P(d(q_F))$ of the accepting state is $1$.
In addition, $P(1) - P(d_{\texttt{max}}) = \kappa$.
The $\kappa$-shaped reward function 
$\Rwd_\kappa: Q \times \alphabet \times Q \to \Real$ is defined as the difference between potentials of the destination and of the target states of transition of the reward machines, \emph{i.e.,} $\Rwd_\kappa((x, q), v, (x', q')) = P(d(q')) - P(d(q))$. 
Moreover, for every run $r = (x_0, q_0), v_1, (x_1, q_1), \ldots,  (x_n, q_n)$ of $\Gg\times{\mathcal A}_\Rwd$, its accumulated reward is simply the potential difference
between the last and the first states, \emph{i.e.,} $P(d(q_n)) - P(d(q_0))$.

\begin{theorem}[Correctness of Reward Shaping]{theorem}{shapingthm}
	\label{reward-shaping}
	For every product stochastic game $\Gg\times{\mathcal A}_\Rwd$  with initial state $(x_0, q_0)$ and reward function $\Rwd$, there exists $\kappa_\star > 0$ such that for all $\kappa < \kappa_\star$ the set of optimal expected reward policies for both players is the same as the set of optimal expected reward policies for $\Gg\times{\mathcal A}_\Rwd$ with reward function $\Rwd_\kappa$.
\end{theorem}

Theorem~\ref{reward-shaping} demonstrates one way to shape rewards such that the
optimal policy remains unaffected while making the rewards less sparse.
Along similar lines, one can construct a variety of potential functions and
corresponding shaped rewards with similar correctness properties.

\subsection{Accelerating RL with Multi-Level Discretization}\label{subsec:dis}

The efficiency of tabular RL algorithms depends on the size of the state set of the finite abstraction---with larger state sets typically requiring longer training times. For instance, if two different agents are trained on a coarse discretization and a fine discretization of the same system, then the former will typically have shorter training times at the cost of higher discretization error. However, since the underlying continuous-space system is the same, the two agents will learn roughly similar policies. There is a large body of literature on adaptive state space partitioning \cite{lee2004adaptive, helms2016dynamic}. 
We propose first training an RL agent on a coarse discretization of the system, and then using the resulting policy to initialize training for a different RL agent on a finer discretization of the same system. By repeating this process with increasingly fine discretization levels we reach the final desired discretization level. Our experimental results demonstrate that this multi-level discretization scheme dramatically accelerates the learning process.

The proposed multi-level discretization algorithms begins by creating a coarse discretization of the continuous-space system. 
It then trains a minimax-Q learning agent for a fixed time. 
After this time has elapsed, it decreases the discretization parameter to create a more finely discretized system by, for instance, halving the discretization parameter. It then creates a new learning agent for the more finely discretized problem.
The new Q-values are initialized from the old Q-values where the value of a state-action pair inherits its values from the corresponding nearest state-action pair in the previous discretization.
This new agent is then trained for a fixed time and is used to initialize a new agent on an even finer discretization of the system. This process get repeated until we reach the final desired discretization level. Note that as long as we stop refining the discretization at some point, we retain the convergence guarantees of minimax-Q learning since it converges from arbitrarily initializations.

\section{Experimental Results}
\label{example}
We consider two case studies. The first case study (room temperature network) concerns the temperature regulation in a network of $N= 20$ rooms with a circular topology. The detailed dynamics for this case study is given in the Section~\ref{sec:case}.
We employ Theorem~\ref{convergence} and synthesize a controller for $\Sigma$ via its implicit abstracted games $\widehat\Sigma_{\delta_i}$,
so that the controller maintains the temperature of each room in the comfort zone $[17,18]$ for at least $45$ minutes. 

\begin{wrapfigure}{r}{0.35\textwidth}
	\centering
	\includegraphics[width=4cm]{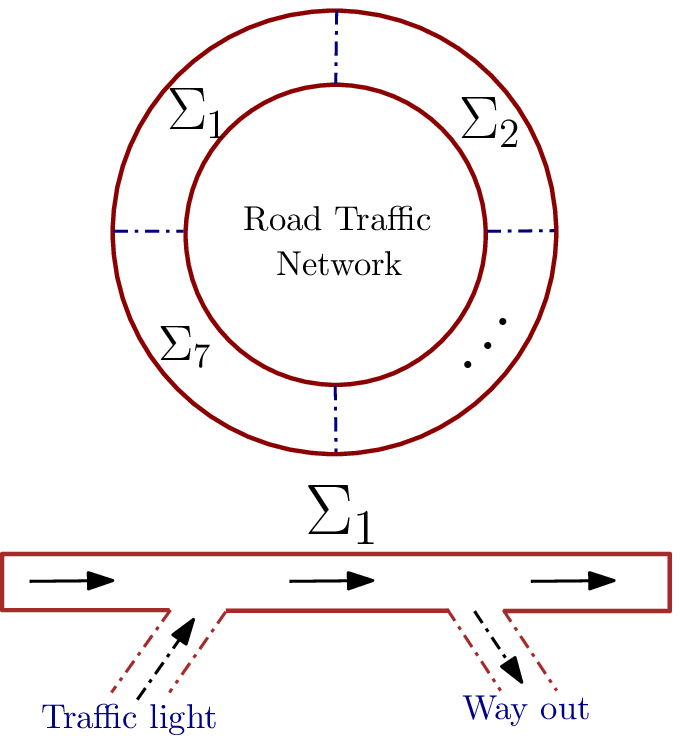}
	\label{Fig2}
\end{wrapfigure}
For the second case study, we consider a road traffic ring network that consists of $N = 7$ identical cells, each of which has $1$ entry and $1$ exit, as shown in the right figure.
The dynamics are given in Section~\ref{sec:case}.
The entry of the cell is controlled by a traffic light, denoted by $v
\in \{0,1\}$, that enables (green light) or not (red light) the vehicles to pass when $v=1$ or $v=0$, respectively. Using Theorem~\ref{convergence}, we synthesize a controller for $\Sigma$ via its implicit abstracted games $\widehat\Sigma_{\delta_i}$ so that the  controller keeps traffic density below $20$ vehicles per cell for at least $36$ seconds. 

\begin{figure}[b!]
	\centering
	\vspace{1em}
	\includegraphics[scale=0.3]{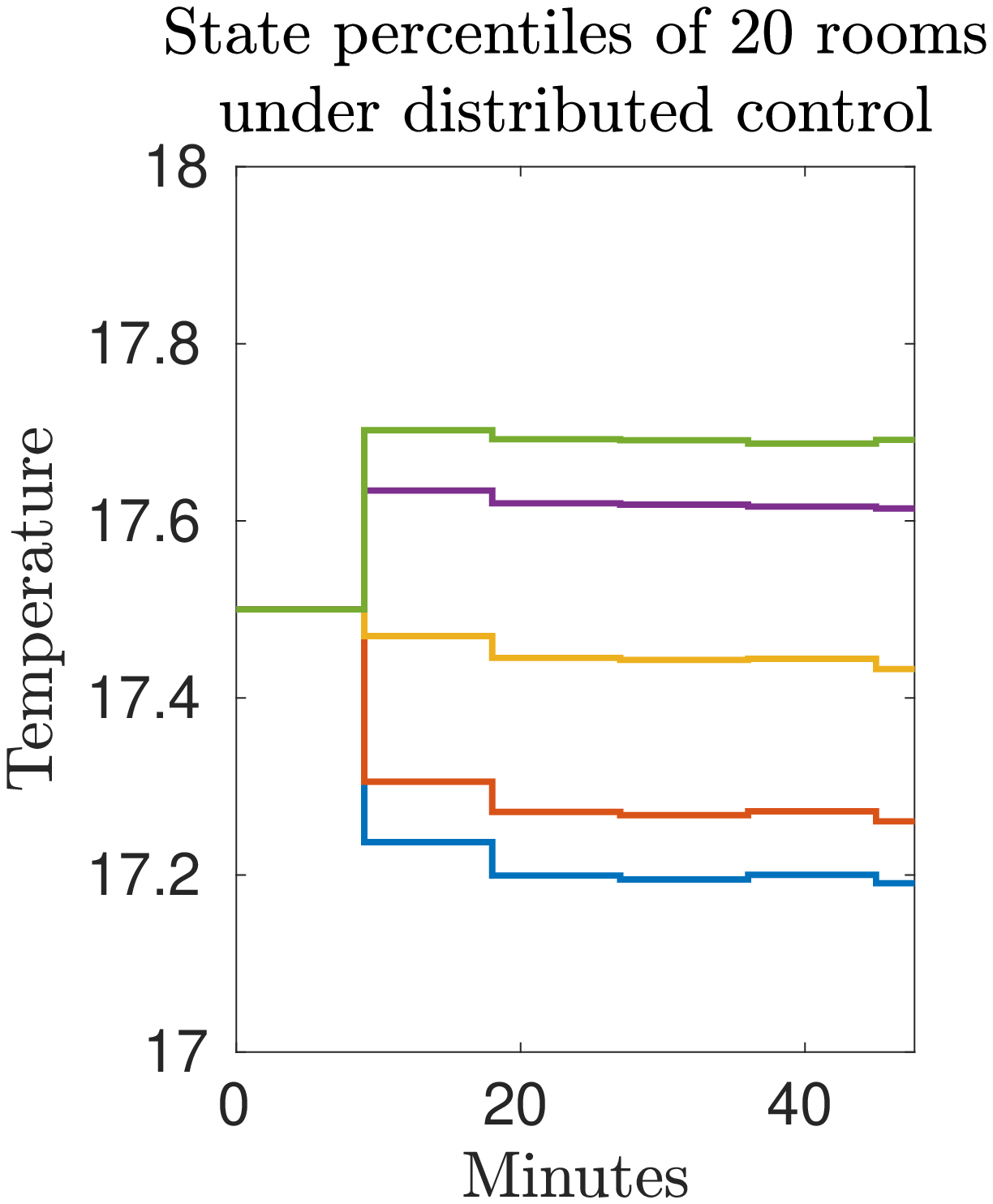}
	\hspace{5em}
	\includegraphics[scale=0.3]{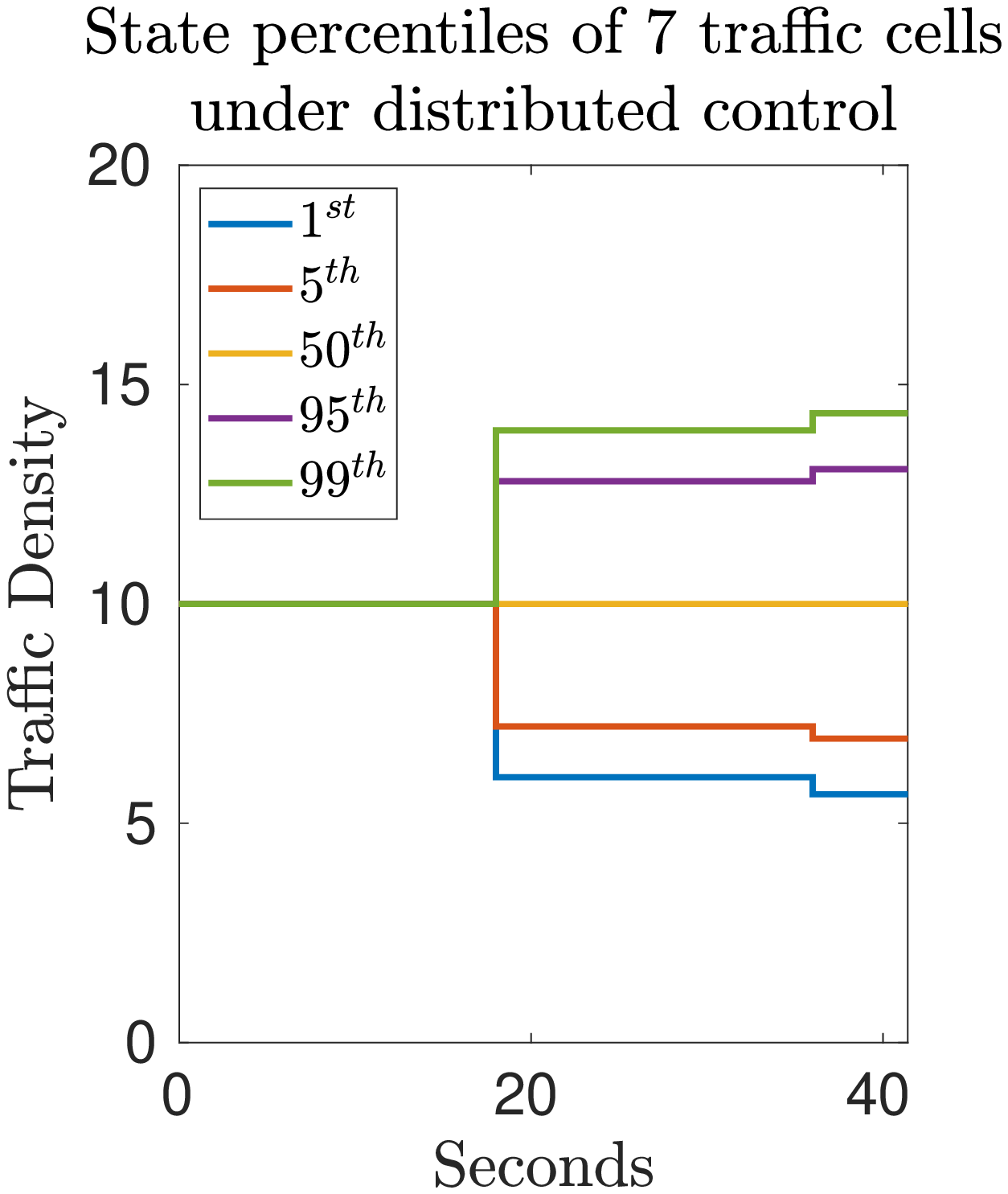}
	\caption{State evolution of the learned distributed controllers visualized through percentiles from $10^6$ sampled trajectories.}
\end{figure}

We synthesize a controller for both case studies by first producing implicitly abstract subsystems.
We then learn a controller for the resulting stochastic games $\widehat\Sigma_{\delta_i}$ with minimax Q-learning. To accelerate learning, we use the multi-level discretization scheme described in Subsection~\ref{subsec:dis}. For the room temperature and road traffic networks, the final discretization values are $\delta_i = 0.001$, $\mu_i = 0.1$, and $\delta_i = 0.05$, $\mu_i = 0.01$, respectively. For the room temperature control, we use 1.5 million episodes, a learning rate of $0.04$ decayed linearly to $0.02$, an exploration rate of $0.1$, and a discount factor of $1$. This takes approximately $5$ minutes of wall-clock time.  For the road traffic network, we employ $2$ million episodes, a learning rate of $0.1$ decayed linearly to $0.02$, an exploration rate of $0.2$, and a discount factor of $1$. This takes approximately $4$ minutes.

Table~\ref{tab:experiments} shows the results for the learned
controllers: $p^+$ is the (approximate) probability of the learned
policy satisfying the finite-horizon objective over the subsystem against an
optimal adversarial internal input, $\varepsilon$ is the bound on the
quantized measurement error from equation~\eqref{eq:metric_lit1},
$p_{\mathrm{low}}$ is the lower bound from
equation~\eqref{eq:metric_com} on the probability of the decentralized
controller satisfying the finite-horizon objective over the interconnected system, and $p_{\mathrm{sampled}}$
is a $95\%$ confidence bound on the probability of satisfying the
finite-horizon objective using the decentralized controller as computed via
$10^6$ samples. We compute $p^+$ in two ways. First, we
approximate $p^+$ without knowledge of the model by fixing the controller policy
that results from learning and producing a $95\%$ confidence bound on the probability of satisfying the objective from $10^6$ samples. Second, we fix the controller policy that results from learning
and compute an optimal strategy for the internal input by dynamic programming. This requires knowledge of the model and is done to validate the results of learning.

\begin{table}[t!]
	\footnotesize
	\caption{\small Results for distributed controller learned by minimax Q-learning on the quantized subsystems.}
	\centering
	\begin{tabular}{ c | c | c | c | c}
		& 
		$p^+$ &
		$p_{\mathrm{low}}$ &
		$\varepsilon$ &
		$p_{\mathrm{sampled}}$
		\\ \hline
		
		\vspace{0.2cm} \textbf{Room}
		& 0.999943   0.999952  $\pm$~0.000014
		& 0.902585   0.902769  $\pm$~0.000272
		& \vspace{0.2cm} 0.004807
		& \vspace{0.01cm} 0.998880  $\pm$~0.000066
		\\ \hline																
		
		\vspace{0.2cm} \textbf{Traffic}
		& 0.996837  0.998878  $\pm$~0.000066
		& 0.932064  0.946167  $\pm$~0.000459
		& \vspace{0.2cm} 0.006571
		& \vspace{0.01cm} 0.999999  $\pm$~0.000002
	\end{tabular}\vspace{0.2cm}
	\label{tab:experiments}
\end{table}

Table~\ref{tab:experiments} shows that on these case studies, the
computed bound on the probability of the decentralized controller
satisfying the finite-horizon objective is within $0.1$ of the
estimated probability using samples for both examples. Additionally, there is a
successful mitigation of the curse of dimensionality versus synthesizing
a centralized controller for the
interconnected system monolithically.
On the room temperature example, the selected quantization parameters result in $n_x = 1000$ states and $n_w = 20$ internal inputs for each
subsystem. Combined with $n_\upsilon = 6$ control inputs and a time horizon
of $\mathcal T=5$, there are $\mathcal T (n_x n_\upsilon + n_x n_\upsilon n_w) = 630,\!000$ total
state-input pairs in the stochastic game which we need to solve to produce
the decentralized controller. For comparison, we can select appropriate quantization parameters which yield the same quantization error in the compositional case, $\frac{1}{2}[(1+\varepsilon)^N\!-\!(1-\varepsilon)^N]$, as in the monolithic case, $\varepsilon$.
The appropriate quantization results in $\hat n_x=9$ states for each individual subsystem. Even with only a few states required in each quantized subsystem, the monolithic approach still needs to reason over $\mathcal T (\hat n_x n_\upsilon)^{20} \approx 2.22\cdot 10^{35}$ state-input pairs to produce a controller. For the road traffic example, there are $n_x = 400$ states, $n_w = 2000$ internal inputs, $n_\upsilon = 2$ control inputs, and a time horizon of $\mathcal{T} = 2$. We require $\hat n_x = 21$ states per subsystem for producing a monolithic controller with the same quantization error as in the compositional case. We get that there are $\mathcal{T}(n_x n_\upsilon + n_x n_\upsilon n_w) = 3,\!201,\!600$ state-input pairs in the stochastic game which we need to solve
s	to produce the decentralized controller, and $\mathcal{T} (\hat n_x n_\upsilon)^{7} \approx 4.61 \cdot 10^{11}$ state-input pairs in the monolithic setting.

\section{Conclusion}
We proposed a \emph{compositional approach} for the policy synthesis of \emph{networks} of unknown stochastic systems using minimax-Q RL. The goal of the policy is to maximize the probability that the system satisfies a logical property. We proposed a lower bound for the probability of satisfaction of a finite-horizon property by the interconnected system based on those of subsystems.
Since automata-based rewards tend to be sparse, we combined our approach with a potential-based \emph{reward shaping} technique and a multi-level discretization to speed up the learning procedure. We demonstrated the effectiveness of the proposed approach by designing the control for two physical case studies. 

\bibliographystyle{alpha}
\bibliography{biblio}

\newcommand{\etalchar}[1]{$^{#1}$}
\begin{thebibliography}{COMB19}

\bibitem[BK08]{baier2008principles}
C.~Baier and J.-P. Katoen.
\newblock {\em Principles of model checking}.
\newblock MIT press, 2008.

\bibitem[COMB19]{cheng2019end}
Richard Cheng, G{\'a}bor Orosz, Richard~M Murray, and Joel~W Burdick.
\newblock End-to-end safe reinforcement learning through barrier functions for
  safety-critical continuous control tasks.
\newblock In {\em AAAI Conference on Artificial Intelligence}, volume~33, pages
  3387--3395, 2019.

\bibitem[CY95]{Courco95}
C.~Courcoubetis and M.~Yannakakis.
\newblock The complexity of probabilistic verification.
\newblock {\em J. ACM}, 42(4):857--907, 1995.

\bibitem[DKS{\etalchar{+}}13]{ding2013stochastic}
Jerry Ding, Maryam Kamgarpour, Sean Summers, Alessandro Abate, John Lygeros,
  and Claire Tomlin.
\newblock A stochastic games framework for verification and control of discrete
  time stochastic hybrid systems.
\newblock {\em Automatica}, 49(9):2665--2674, 2013.

\bibitem[FV97]{Filar97}
J.~Filar and K.~Vrieze.
\newblock {\em Competitive {Markov} Decision Processes}.
\newblock Springer, 1997.

\bibitem[HKA{\etalchar{+}}19]{hasanbeig2019reinforcement}
M.~Hasanbeig, Y.~Kantaros, A.~Abate, D.~Kroening, G.~J. Pappas, and I.~Lee.
\newblock Reinforcement learning for temporal logic control synthesis with
  probabilistic satisfaction guarantees.
\newblock In {\em CDC}, pages 5338--5343, 2019.

\bibitem[HMU16]{helms2016dynamic}
T.~Helms, S.~Mentel, and A.~Uhrmacher.
\newblock Dynamic state space partitioning for adaptive simulation algorithms.
\newblock In {\em Proceedings of the 9th EAI International Conference on
  Performance Evaluation Methodologies and Tools}, pages 149--152, 2016.

\bibitem[HS20]{HS18_robust}
Sofie Haesaert and Sadegh Soudjani.
\newblock Robust dynamic programming for temporal logic control of stochastic
  systems.
\newblock {\em IEEE Transactions on Automatic Control}, 66(6):2496--2511, 2020.

\bibitem[IKVM18]{icarte2018using}
Rodrigo~Toro Icarte, Toryn Klassen, Richard Valenzano, and Sheila McIlraith.
\newblock Using reward machines for high-level task specification and
  decomposition in reinforcement learning.
\newblock In {\em International Conference on Machine Learning}, pages
  2107--2116. PMLR, 2018.

\bibitem[JPZ20]{jagtap2020control}
P.~Jagtap, G.~J. Pappas, and M.~Zamani.
\newblock Control barrier functions for unknown nonlinear systems using
  {Gaussian} processes.
\newblock In {\em CDC}, pages 3699--3704, 2020.

\bibitem[KS20]{kazemi2020formal}
M.~Kazemi and S.~Soudjani.
\newblock Formal policy synthesis for continuous-state systems via
  reinforcement learning.
\newblock In {\em International Conference on Integrated Formal Methods}, pages
  3--21. Springer, 2020.

\bibitem[KV01]{KupfermanVardi2001}
O~Kupferman and M~Y Vardi.
\newblock Model checking of safety properties.
\newblock {\em Formal Methods in System Design}, 19(3):291--314, 2001.

\bibitem[Lav19]{lavaei2019Thesis}
A.~Lavaei.
\newblock {\em Automated Verification and Control of Large-Scale Stochastic
  Cyber-Physical Systems: Compositional Techniques}.
\newblock PhD thesis, sTechnische Universit{\"a}t M{\"u}nchen, Germany, 2019.

\bibitem[Lit94]{Littma94}
M.~L. Littman.
\newblock Markov games as a framework for multi-agent reinforcement learning.
\newblock In {\em International Conference on Machine Learning}, pages
  157--163, 1994.

\bibitem[LKZ21]{luo2021abstraction}
X.~Luo, Y.~Kantaros, and M.M. Zavlanos.
\newblock An abstraction-free method for multirobot temporal logic optimal
  control synthesis.
\newblock {\em IEEE Transactions on Robotics}, 2021.

\bibitem[LL04]{lee2004adaptive}
I.S.K. Lee and H.Y.K. Lau.
\newblock Adaptive state space partitioning for reinforcement learning.
\newblock {\em Engineering applications of artificial intelligence},
  17(6):577--588, 2004.

\bibitem[LMF{\etalchar{+}}16]{lahijanian2016iterative}
M.~Lahijanian, M.R. Maly, D.~Fried, L.E. Kavraki, H.~Kress-Gazit, and M.Y.
  Vardi.
\newblock Iterative temporal planning in uncertain environments with partial
  satisfaction guarantees.
\newblock {\em IEEE Transactions on Robotics}, 32(3):583--599, 2016.

\bibitem[LS96]{Littma96}
M.~L. Littman and C.~Szepesvari.
\newblock A generalized reinforcement-learning model: Convergence and
  applications.
\newblock In {\em International Conference on Machine Learning}, pages
  310--318, 1996.

\bibitem[LSAZ22]{Lavaei_Survey}
A.~Lavaei, S.~Soudjani, A.~Abate, and M.~Zamani.
\newblock Automated verification and synthesis of stochastic hybrid systems: A
  survey.
\newblock {\em Automatica}, 2022.

\bibitem[LSS{\etalchar{+}}20]{lavaei2020ICCPS}
A.~Lavaei, F.~Somenzi, S.~Soudjani, A.~Trivedi, and M.~Zamani.
\newblock Formal controller synthesis for continuous-space {MDP}s via
  model-free reinforcement learning.
\newblock In {\em International Conference on Cyber-Physical Systems (ICCPS)},
  pages 98--107, 2020.

\bibitem[LSZ19]{lavaei2018CDCJ}
A.~Lavaei, S.~{Soudjani}, and M.~Zamani.
\newblock Compositional construction of infinite abstractions for networks of
  stochastic control systems.
\newblock {\em Automatica}, 107:125--137, 2019.

\bibitem[LSZ20]{lavaei2018ADHSJ}
A.~Lavaei, S.~{Soudjani}, and M.~Zamani.
\newblock Compositional (in)finite abstractions for large-scale interconnected
  stochastic systems.
\newblock {\em IEEE Transactions on Automatic Control}, 65(12):5280--5295,
  2020.

\bibitem[LZ22]{Lavaei_TAC2022}
A.~Lavaei and M.~Zamani.
\newblock From dissipativity theory to compositional synthesis of large-scale
  stochastic switched systems.
\newblock {\em IEEE Transactions on Automatic Control}, 2022.

\bibitem[MMSS21]{majumdar2021symbolic}
Rupak Majumdar, Kaushik Mallik, Anne-Kathrin Schmuck, and Sadegh Soudjani.
\newblock Symbolic qualitative control for stochastic systems via finite parity
  games.
\newblock {\em IFAC-PapersOnLine}, 54(5):127--132, 2021.

\bibitem[SA13]{SA13}
S.~{Soudjani} and A.~Abate.
\newblock Adaptive and sequential gridding procedures for the abstraction and
  verification of stochastic processes.
\newblock {\em SIAM J. Applied Dynamical Systems}, 12(2):921--956, 2013.

\bibitem[SAM17]{SAM17}
S.~Soudjani, A.~Abate, and R.~Majumdar.
\newblock Dynamic {B}ayesian networks for formal verification of structured
  stochastic processes.
\newblock {\em Acta Informatica}, 54(2):217--242, 2017.

\bibitem[SB18]{Sutton18}
R.~S. Sutton and A.~G. Barto.
\newblock {\em Reinforcement Learning: An Introduction}.
\newblock MIT Press, second edition, 2018.

\bibitem[{Sou}14]{SSoudjani}
S.~{Soudjani}.
\newblock {\em Formal Abstractions for Automated Verification and Synthesis of
  Stochastic Systems}.
\newblock PhD thesis, Technische Universiteit Delft, 2014.

\bibitem[WD92]{watkins1989learning}
Christopher J. C.~H. Watkins and Peter Dayan.
\newblock Q-learning.
\newblock In {\em Machine Learning}, pages 279--292, 1992.

\bibitem[WLSL20]{wang2020continuous}
C.~Wang, Y.~Li, S.L. Smith, and J.~Liu.
\newblock Continuous motion planning with temporal logic specifications using
  deep neural networks.
\newblock {\em arXiv:2004.02610}, 2020.

\bibitem[YHAK19]{yuan2019modular}
L.~Z. Yuan, M.~Hasanbeig, A.~Abate, and D.l Kroening.
\newblock Modular deep reinforcement learning with temporal logic
  specifications.
\newblock {\em arXiv:1909.11591}, 2019.

\end{thebibliography}

\newpage 

\section{Appendix}

\subsection{Dynamics of Case Studies}
\label{sec:case}

\noindent{\bf Room Temperature Network:} The evolution of the temperatures can be described by an interconnected \mbox{dt-SCS} as  
\begin{align}\label{Room_Nerwork}
\Sigma:{x}(k+1)=A{T_x}(k)+\gamma T_{h}\upsilon(k)+ \beta T_{E}+R\varsigma(k),
\end{align}
where $A$ is a matrix with diagonal elements $a_{ii}=(1-2\psi-\beta-\gamma\upsilon_{i}(k))$, $i\in\{1,\ldots,N\}$, off-diagonal elements $a_{i,i+1}=a_{i+1,i}=a_{1,N}=a_{N,1}=\psi$, $i\in \{1,\ldots,N-1\}$, and all other elements are identically zero.
Parameters $\psi = 0.001$, $\beta = 0.4$, and $\gamma = 0.5$, where $\psi$ is a conduction factor between each two rooms in the network. Moreover, $
T_x(k)=[T_{x_1}(k);\ldots;T_{x_N}(k)]$,
$\upsilon(k)=[\upsilon_1(k);\ldots;\upsilon_N(k)]$, $
\varsigma(k)=[\varsigma_1(k);\ldots;\varsigma_N(k)]$,
$T_E=[T_{e1};\ldots;T_{eN}]$, where $T_i(k)$ takes values in the set $[17,18]$ and $\upsilon_i(k)$ takes values in the finite input set $\{1.1542, 1.1625, 1.1708, 1.1792, 1.1875, 1.1958\}$ made from the center-points of $6$ equal partitions of the interval $[1.15,1.20]$, for all $i\in\{1,\ldots,N\}$.
Furthermore, $R$ is a diagonal matrix with elements $r_{ii} = 0.1, i\in\{1,\dots,N\}$. Outside temperatures are the same for all rooms: $T_{ei}=-1\,^\circ C$, $\forall i\in\{1,\ldots,N\}$, and the heater temperature $T_h=30\,^\circ C$. Now, by introducing $\Sigma_i$ described as
\begin{equation*}
\Sigma_i\!:\!\left\{\hspace{-1.5mm}\begin{array}{l}T_{x_i}(k\!+\!1)=a_{ii}T_{x_i}(k)+\gamma T_{h} \upsilon_i(k)+D_i w_i(k)+\beta T_{ei}+0.1\varsigma_i(k),\\
y_i(k)=T_{x_i}(k),\\
\end{array}\right.
\end{equation*}
one can verify that $\Sigma=\mathcal{I}_g(\Sigma_1,\ldots,\Sigma_N)$ where $D_i = [\psi;\psi]^T$, and $w_i(k) = [y_{i-1}(k);y_{i+1}(k)]$ (with $y_0 = y_N$ and $y_{n+1} = y_1$).

\vspace{0.3cm}
\noindent{\bf Road Traffic Network:} We consider the length of each cell as $0.5$ kilometers ($[km]$), and the
flow speed of the vehicles as $45$ kilometers per hour
($[km/h]$). Moreover, during the sampling time interval $\tau=18$
seconds, it is assumed that $5$ vehicles pass the entry controlled by
the green light,  and half the vehicles exit each cell (ratio denoted
by $\bar q$).
We want to observe traffic density $x_i$ (vehicles per cell) for each cell $i$ of the road. The dynamic of the interconnected system is described by
\begin{align}\label{Road_Network}
\Sigma:x(k+1) = Ax(k)+ \bar B\upsilon(k)+ R\varsigma(k)+G,
\end{align}
where $A$ is a matrix with diagonal elements $ a_{ii}=(1-\frac{\tau \bar\nu_i}{l_i}-\bar q)$, $i\in\{1,\ldots,N\}$, off-diagonal elements $ a_{i+1,i}=\frac{\tau \bar\nu_i}{l_i}$, $i\in \{1,\ldots,N-1\}$, $a_{1,N} = \frac{\tau \bar\nu_N}{l_N}$, and all other elements are identically zero. Moreover, $\bar B$ and $R$ are diagonal matrices with elements $b_{ii} = 5$, and $r_{ii} = 1.7, i\in\{1,\dots,N\}$, respectively, and $G = \mathbf{0}_N$. Furthermore, $ x(k)=[x_1(k);\ldots;x_{N}(k)]$,  $\upsilon(k)=[\upsilon_1(k);\ldots;\upsilon_{N}(k)]$, and $ \varsigma(k)=[\varsigma_1(k);\ldots;\varsigma_{N}(k)]$.
Now by introducing individual cells $\Sigma_i$ described as
\begin{align}\notag
\Sigma_i\!:\!\left\{\hspace{-1.5mm}\begin{array}{l}x_i(k+1) = (1-\frac{\tau \bar\nu_i}{l_i}-q)\,x_i(k) + D_iw_i(k) +5\upsilon_i(k)+ 1.7 \varsigma_i(k),\\
y_i(k)=x_i(k),\\
\end{array}\right.
\end{align}
one can readily verify that $\Sigma=\mathcal{I}_g(\Sigma_1,\ldots,\Sigma_N)$ where
$D_i = \frac{\tau \bar\nu_{i-1}}{l_{i-1}}$ (with $\bar\nu_0 = \bar\nu_N$, $l_0 =
l_N$) and $w_i(k) = y_{i-1}(k)$ (with $y_0 = y_N$).
Observe that 
$X_i = W_i = [0~~20]$ and $\upsilon_i(k) \in \{0, 1\}$ for $i\in\{1,\dots,N\}$.

\end{document}